\documentclass[12pt,reqno]{amsart}
\usepackage[headings]{fullpage}
\usepackage{amssymb,amsmath,amscd,bbm,tikz}
\usepackage[all,cmtip]{xy}
\usepackage{url}
\usepackage[bookmarks=true,%
    colorlinks=true,%
    linkcolor=blue,%
    citecolor=blue,%
    filecolor=blue,%
    menucolor=blue,%
    urlcolor=blue,%
    breaklinks=true]{hyperref}
\usepackage{slashed}    
\usepackage{verbatim}

\newtheorem{theorem}{Theorem}[section]
\theoremstyle{definition}
\newtheorem{proposition}[theorem]{Proposition}
\newtheorem{lemma}[theorem]{Lemma}

\newtheorem{corollary}[theorem]{Corollary}

\def\BN{\mathbbm N}
\def\BZ{\mathbbm Z}

\def\BR{\mathbbm R}
\def\BC{\mathbbm C}

\def\calL{\mathcal L}

\def\calB{\mathcal B}

\def\pt{\partial}

\def\tq{\tilde{q}}
\def\Res{\mathrm{Res}}

\def\d{\delta}
\def\e{\epsilon}

\def\th{\theta}

\def\Li{\mathrm{Li}}

\def\Re{\mathrm{Re}}
\def\Im{\mathrm{Im}}

\def\be{\begin{equation}}
\def\ee{\end{equation}}

\newcommand{\im}{\mathsf{i}}
\newcommand{\fad}[1]{\operatorname{\Phi}_{#1}}
\newcommand{\fadb}{\operatorname{\Phi}_{\mathsf{b}}}

\def\sfb{\mathsf{b}}
\def\phih{\hat{\phi}}
\def\tq{\tilde{q}}

\def\dd{\operatorname{d}}
\begin{document}
\title[Resurgence of Faddeev's quantum dilogarithm]{
  Resurgence of Faddeev's quantum dilogarithm}
\author{Stavros Garoufalidis}
\address{
  International Center for Mathematics, Department of Mathematics \\
  Southern University of Science and Technology \\
  Shenzhen, China \newline
  {\tt \url{http://people.mpim-bonn.mpg.de/stavros}}}
\email{stavros@mpim-bonn.mpg.de}
\author{Rinat Kashaev}
\address{Section de Math\'ematiques, Universit\'e de Gen\`eve \\
2-4 rue du Li\`evre, Case Postale 64, 1211 Gen\`eve 4, Switzerland \newline
         {\tt \url{http://www.unige.ch/math/folks/kashaev}}}
\email{Rinat.Kashaev@unige.ch}

\date{4 October, 2020}

\begin{abstract}
  The quantum dilogarithm function of Faddeev is a special function that
  plays a key role as the building block
  of quantum invariants of knots and 3-manifolds,
  of quantum Teichm\"uller theory and of complex Chern--Simons theory.
  Motivated by conjectures on resurgence and the recent interest in
  wall-crossing phenomena, we
  prove that the Borel summation of a formal power series solution
  of a linear difference equation produces Faddeev's quantum dilogarithm.
  Along the way, we give an explicit formula for the Borel transform,
  a meromorphic function in the Borel plane, locate its poles and residues
  and describe the Stokes
  phenomenon of its Laplace transforms along the Stokes rays.
\end{abstract}

\maketitle

{\footnotesize
\tableofcontents
}



\section{Introduction}
\label{sec.intro}

A well-known problem in quantum topology is the Volume Conjecture which
asserts that the Kashaev invariant of a hyperbolic knot grows exponentially
at a rate proportional to the volume of the knot~\cite{K94,K95,K97}. There are
several strengthenings of this conjecture that involve the analytic
properties of the asymptotics of the Kashaev invariant to all orders
(see e.g.,~\cite{DGLZ,GZ:kashaev} and references therein). Such
factorially divergent formal power series have been conjectured to lead to
resurgent functions~\cite{Ga:resurgence}, and this in turn leads to 
astonishing numerically testable
conjectures~\cite{Gukov:resurgence,GZ:kashaev,GGM:resurgent}.
The Kashaev invariant of a knot
is a finite state-sum whose building block is the quantum $n$ factorial
$(q;q)_n=\prod_{j=1}^n (1-q^j)$, evaluated at complex roots of unity.
The latter is intimately related to another special function, the
Faddeev quantum dilogarithm~\cite{Faddeev}, evaluated at rational points.
Although the conjectured resurgence properties of quantum knot invariants are
largely unproven, in an unfinished manuscript from 2006 we studied the
resurgence properties of their building block, namely the Faddeev quantum
dilogarithm. This special function plays a key role in quantum
Teichm\"uller theory~\cite{Kashaev:qteichmuller,AK:TQFT} and complex Chern--Simons
theory~\cite{Beem:holomorphic-blocks,DGG1,Dimofte:3d}. 
Since there is renewed interest in this subject with applications
to resurgence and wall-crossing phenomena (see for instance~\cite{MaximTalks}),
we decided to update our manuscript and make it widely available.

To begin the story, in the quantization of Teichm\"uller theory one considers
the difference equation
\be
\label{eq.diff}
f_\tau(z-i \pi \tau) = (1+e^{z})f_\tau(z+i \pi \tau) 
\ee
whose motivation is explained in detail in~\cite[Prop.8]{Kashaev:qteichmuller}
and also in~\cite[Prop.1, Eqn.(9)]{AK:TQFT} (after some minor change in
notation). The above difference equation
appears, among other places, in quantum integrable systems
(see Ruijsenaars~\cite[Eqn.(1.17)]{Rui}) and in holomorphic dynamics
(see Marmi--Sauzin~\cite{MS}). 

It it easy to see (see Lemma~\ref{lem.formal} below) that if
$f_\tau(z)$ satisfies Equation~\eqref{eq.diff} and the limiting value
$\lim_{z \to -\infty}f_\tau(z)=1$, then for a fixed $z$, $f_\tau(z)$
admits an asymptotic expansion of the form
\be
\label{eq.Phia}
\log f_\tau(z) \sim \frac{1}{2\pi i\tau} \Li_2(-e^z) +
\phih_\tau(z), \qquad (\tau \to 0)
\ee
where
\be
\label{eq.phit}
\phih_\tau(z) =
\sum_{n=1}^\infty (2 \pi i)^{2n-1} \frac{B_{2n}(1/2)}{(2n)!} 
\pt_z^{2n} \Li_{2}(-e^z) \tau^{2n-1} \,,
\ee
$\Li_2(z)=\sum_{k \geq 1} z^k/k^2$ is Euler's dilogarithm function and the
differentiation operator $\pt_z^{2n}$ (defined by $\pt_z g(z)=g'(z)$) acts
on $\Li_{2}(-e^z)$.

The goal of this paper is to identify the Borel summation of
the factorially divergent series $\hat \phi_\tau(z)$ with Faddeev's quantum
dilogarithm function
\be
\label{eq.fff}
f_\tau(z)=\fadb(z/(2 \pi \sfb))
\ee
(see Corollary~\ref{cor.1} below) when $\tau=\sfb^2 >0$. Along the way,
we give an explicit formula for the Borel transform $G(\xi,z)$ of the
power series $\hat \phi_\tau(z)$ (see Theorem~\ref{thm.1} below).

It turns out that $G(\xi,z)$ is a meromorphic function of $\xi$ with
poles that lie discretely in a countable union $L(z)$ of lines through
the origin given in Equation~\eqref{eq.L} below. The arrangement $L(z)$
depends on $z$ and accumulates to the imaginary axis. Such an arrangement is
reminiscent of the parametric resurgence of non-linear equations (see
e.g.~\cite{Sauzin:divergent}), the exact and perturbative invariants of
Chern--Simons theory (predicted for instance in~\cite{Ga:resurgence} and
Figure below Defn.~2.3 of ibid), and the wall-crossing formulas of
Kontsevich--Soibelman (see~\cite{Kontsevich-Soibelman:coho} and
also~\cite{MaximTalks}). 

Theorem~\ref{thm.2} identifies the Laplace transform of $G(\cdot,z)$ with
the logarithm of Faddeev's quantum dilogarithm function $f_\tau(z)$ given
in Equation~\eqref{eq.fff}. We also consider the Laplace transform of
the function $G(\cdot,z)$ along any ray in the complement of $L(z)$
and describe the Stokes phenomenon,
i.e., the change of the Laplace transform as one crosses a Stokes line
$\BR \xi_m(z)$. One may think of this as an instance of a wall-crossing
formula, in the spirit of Kontsevich--Soibelman.

Another noteworthy phenomenon is the Laplace transform of $G(\cdot,z)$
along the vertical rays $\pm i \BR_+$, which no longer lie in an open
cone in the complement of $L(z)$. This is the case considered by
Marmi--Sauzin~\cite{MS} who prove (with a careful analysis) that the Laplace
transform $f^-_\tau$ (resp. $f^+_\tau$) is defined in the upper half-plane
$\Im(\tau)>0$ (resp. lower half-plane $\Im(\tau)<0$) thus leading to two
distinguished solutions $f^\pm_\tau(z)$ of Equation~\eqref{eq.diff}.


\begin{figure}[!htpb]
\begin{tikzpicture}[scale=0.7]
\draw[fill=black] (1, - 2.17911 ) circle (1.5pt)
(1, - 1.5508)  circle (1.5pt)
(1, - .922478)  circle (1.5pt)
(1, - .294159) circle (1.5pt)
(1, .334159) circle (1.5pt)
(1, .962478)  circle (1.5pt)
(1, 1.5908)  circle (1.5pt)
(1, 2.21911)  circle (1.5pt)
(1, 2.84743)  circle (1.5pt)

(-1,  2.17911 ) circle (1.5pt)
(-1,  1.5508)  circle (1.5pt)
(-1,  .922478)  circle (1.5pt)
(-1,  .294159) circle (1.5pt)
(-1, -.334159) circle (1.5pt)
(-1,-.962478)  circle (1.5pt)
(-1, -1.5908)  circle (1.5pt)
(-1, -2.21911)  circle (1.5pt)
(-1, -2.84743)  circle (1.5pt)

(2., - 1.84496) circle (1.5pt)
(2., - .588319) circle (1.5pt)
(2.,.668319) circle (1.5pt)
(2.,1.92496) circle (1.5pt)
(2.,3.18159) circle (1.5pt)

(-2., 1.84496) circle (1.5pt)
(-2.,  .588319) circle (1.5pt)
(-2.,-.668319) circle (1.5pt)
(-2.,-1.92496) circle (1.5pt)
(-2.,-3.18159) circle (1.5pt)
;
 \draw[->,thick, color=blue] (-4,0)--(4,0);
 \draw[->,thick, color=blue]  (0,-4.1)--(0,4.1);
  \draw [shorten >= -.7cm, shorten <=-.7cm]  (1, - 2.17911 ) -- (-1, 2.17911 );
 \draw [shorten >= -.7cm, shorten <=-.7cm] (1, - 1.5508) --(-1,  1.5508) ;
 \draw [shorten >= -.7cm, shorten <=-.7cm] (2., - 1.84496)--(-2., 1.84496);
 \draw [shorten >= -.7cm, shorten <=-.7cm] (2., - .588319)--(-2.,  .588319);
\draw [shorten >= -.7cm, shorten <=-.7cm] (2.,.668319)--(-2.,-.668319);
\draw [shorten >= -.7cm, shorten <=-.7cm] (2.,1.92496)--(-2.,-1.92496);
\draw [shorten >= -.7cm, shorten <=-.7cm] (2.,3.18159) --(-2.,-3.18159) ;
 \draw [shorten >= -.7cm, shorten <=-.7cm] (1, 2.21911)  --(-1, -2.21911) ;
 \draw [shorten >= -.7cm, shorten <=-.7cm] (1, 2.84743) --(-1, -2.84743);
\end{tikzpicture}
\begin{center}
  \caption{The poles of $G(\xi,z)$ in the $\xi$-plane are points $n\xi_m(z)$
    lie in an arrangement of lines passing through the origin. }
\label{f.Gpoles}
\end{center}
\end{figure}
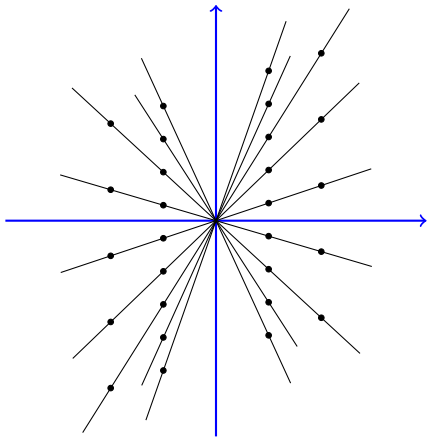

\subsection{Our results}
\label{sub.results}

Recall the quantum dilogarithm function of Faddeev~\cite{Faddeev}
\be
\label{eq.fad}
\fadb(z) =\exp\left( \int_{\mathbb{R}+\im\epsilon}
\frac{e^{-2\im xz}}{4\sinh(x\sfb)\sinh(x\sfb^{-1})}
\frac{\operatorname{d}\!x}{x} \right) \,,
\ee
a function with remarkable analytic properties that satisfies
a pentagon identity and an inversion relation summarized in Section~\ref{sec.fad}
below. $\fadb(z)$ is a meromorphic, quasi-periodic function of $z$
that satisfies
\be
\label{eq.diff0}
\fadb(z-i \sfb/2) =
(1+e^{2 \pi \sfb z})\fadb(z+i \sfb /2) \,.
\ee
The function $\fadb(z)$ is used as the
building block of topological invariants of 3-manifolds via quantum
Teichm\"uller theory~\cite{AK:TQFT,AK:complex,AK:new,KLV}. 

Consider the Borel transform
\be
\label{eq.borel}
\calB: \tau \BC[[\tau]] \to \BC[[\xi]], \qquad
\calB(\tau^{n+1})=\frac{\xi^n}{n!} \,.
\ee
Let 
\be\label{eq.functionG}
G(\xi,z) = \calB(\phih_\tau(z))
\ee
denote the Borel transform of $\phih_\tau(z)$. Our first result describes
a global formula for $G(\xi,z)$ in the complex Borel $\xi$-plane.

\begin{theorem}
  \label{thm.1}
  When $z \in \BC$ with $|\Im(z)| < \pi$ and $\xi \in \BC$ with
  $|\xi|<\pi-|\Im(z)|$, we have:
  \be
  \label{eq.G}
  G(\xi,z)=\frac{1}{2\pi i} \sum_{n=1}^\infty \frac{(-1)^n}{n^2}
  \left(\frac{1}{1+e^{\frac\xi n-z}} + \frac{1}{1+e^{-\frac\xi n-z}} \right) \,,
  \ee
  where the right hand side is expanded as a formal power series in $\xi$
  around zero with radius of convergence $\pi-|\Im(z)|$. \newline
  It follows that $G(\xi,z)$ is a meromorphic function of $\xi$ with simple
  poles at $\xi=n \xi_m(z)$ (shown in Figure~\ref{f.Gpoles}) with residue
  $C_n$, where
  \be
  \label{eq.xiab}
  \xi_m(z)=z + (2m+1)\pi i, \qquad n \in \BZ^*=\BZ\setminus\{0\},
  \quad m \in \BZ, \qquad C_n = \frac{(-1)^n}{2 \pi i n} \,.
  \ee
\end{theorem}

\noindent
Note that the singularities of $G(\cdot,z)$ form a lattice in a countable
union $L(z)$ of lines through the origin
\be
\label{eq.L}
L(z)=\cup_{m \in \BZ} \BR \xi_m(z) \,.
\ee
The formula of the above theorem is similar to the one of
Marmi--Sauzin~\cite[Thm.~4.3]{MS}. It is also similar to the Euler--MacLaurin
summation formula given in Costin--Garoufalidis~\cite[Thm.~2]{CG:EM}. This
is not an accident. When $\Im(\tau)>0$, the quantum dilogarithm has an
infinite product expansion (see Equation~\eqref{eq.phiproduct} below)
whose logarithm can be written as an infinite sum which one
can analyze with the Euler--MacLaurin summation method. However, such a
manipulation is meaningless since the product formula~\eqref{eq.phiproduct},
although convergent when $\Im(\tau)>0$, diverges when $\tau>0$.

Recall the Laplace transform
\be
\label{eq.laplace}
(\calL f)(\tau) = \int_0^\infty e^{-\xi/\tau} f(\xi) \dd\!\xi
\ee
of a function $f \in L^1(\BR)$. Our next theorem concerns the analytic
properties of the meromorphic function $G(\xi,z)$ and its Laplace
transform $(\calL G)(\tau,z)$. For a positive real
number $\d$, define
\be
S_\d=\{ \xi \in \BC \,\, | \,\, \text{dist}(\xi,[0,\infty))< \d \},
\qquad \begin{tikzpicture}[scale=.5,baseline=-2]
 \draw[gray!20, fill=gray!20] (0,-1) rectangle (3,1);
 \draw[thick,fill=gray!20]
(3,1)-- (0,1)
  (0,1)  arc(90:270:1)
 (0,-1)--(3,-1); 

 \draw
 (0,0)--(3,0);
 \draw[fill=black] (0,0) circle(.05);
 \draw[<->,thin] (2,0)--node[midway,right]{\tiny$\delta$}(2,1);
 \draw(0,0.4) node{\tiny$0$};
 \draw(3.5,-0.0) node{\tiny$S_\delta$};
\end{tikzpicture}
\ee
\begin{theorem}
  \label{thm.anal}
  \rm{(a)}
  When $z \in \BC$ with $0 < \d < \pi - |\Im(z)|$ and $\xi \in S_\d$ we
  have:
\be
  \label{eq.Gbound}
  |G(\xi,z)| \leq \max\left\{2, \frac{\pi^2}{3\sin(|\Im(z)|+\d)} \right\} \,.
  \ee
  \rm{(b)} Fix $0<\d<\pi$. Then we have:
  \be
  \label{eq.lapb}
  \left|(\calL G)(\tau,z)
  - \sum_{n=1}^N (2 \pi i)^{2n-1} \frac{B_{2n}(1/2)}{(2n)!} 
  \pt_z^{2n} \Li_{2}(-e^z) \tau^{2n-1} \right| \leq L M^{2N} (2N)!
  \, \Re(\tau) \, |\tau|^{2N}
  \ee
for all $\tau \in \BC$ with $\Re(\tau)>0$ and all $z \in \BC$ with
  $|\Im(z)| < \pi -\d$ and all natural numbers $N$, where
  $M=M_\d=2/\d$ and $L=L_{z,\d} = \max\left\{2, \frac{\pi^2}{3\sin(|\Im(z)|+\d)}
  \right\}$.
\end{theorem}

Our next result identifies the Borel transform of $G(\xi,z)$ with the
quantum dilogarithm function.

\begin{theorem}
  \label{thm.2}
  When $\tau>0$ and $z \in \BC$ with $|\Im(z)| < \pi$, we have:
  \be
  \label{eq.thm2}
  \log \fadb\left(\frac{z}{2\pi \sfb}\right) = \frac{1}{2\pi i\tau} \Li_2(-e^z)
  + (\calL G)(\tau,z)
  \ee
  where $\sfb^2=\tau$ and $G(\xi,z)$ is as in~\eqref{eq.G}.
  \end{theorem}

  This theorem follows from the explicit formula for $G(\xi,z)$
    in Theorem~\ref{thm.1} which agrees with the integral formula of Woronowicz
  for $\fadb(z)$. Theorems~\ref{thm.anal} and~\ref{thm.2} give the following.

  \begin{corollary}
    \label{cor.est}
With the assumptions of part (b) of Theorem~\ref{thm.anal}, we have:
  \be
  \label{eq.lapc}
  \left|\log \fadb\left(\frac{z}{2\pi \sfb}\right)
  - \sum_{n=0}^N (2 \pi i)^{2n-1} \frac{B_{2n}(1/2)}{(2n)!} 
  \pt_z^{2n} \Li_{2}(-e^z) \tau^{2n-1} \right| \leq L M^{2N} (2N)!
  \, \Re(\tau) \, |\tau|^{2N} \,.
  \ee
\end{corollary}
\noindent
The special case of~\eqref{eq.lapc} with $N=0$ is equivalent to
Lemma 7.13 of~\cite{Ben-Aribi}, which itself is an improvement of
an earlier Lemma 3 of Andersen--Hansen~\cite{Andersen-Hansen}. An alternative proof
of the above inequality~\eqref{eq.lapc} was given by Andersen~\cite{Andersen:private}.
  
  The process of replacing a factorially divergent series with its Borel
  transform, followed by the Laplace transform is known as Borel summation.
  Theorems~\ref{thm.1} and~\ref{thm.2} imply the following.

  \begin{corollary}
    \label{cor.1}
  When $\tau >0$ and $z \in \BC$ with $|\Im(z)| < \pi$,  
  the Borel summation of the series~\eqref{eq.phit} reproduces the logarithm
  of Faddeev's quantum dilogarithm function, namely, $\log \fadb(z/(2\pi \sfb))$.
\end{corollary}

The above corollary has some surprising consequences. A priori, a solution
of~\eqref{eq.diff} is well-defined up to multiplication with
$2\pi i\tau$-periodic functions, and Borel summation chooses exactly
the one that agrees with the quantum dilogarithm function. What's more,
the quantum dilogarithm function satisfies the symmetry of
Equation~\eqref{eq.fsym}, and hence satisfies
a second difference equation (obtained by
replacing $\sfb$ by $\sfb^{-1}$ in~\eqref{eq.diff0}). This,
together with Corollary~\ref{cor.1}, implies the following.

\begin{corollary}
\label{cor.2}
The Borel summation of $\hat \phi_\tau(z)$
satisfies the additional functional equation
\be
\label{eq.diff2}
f_\tau(z-i \pi) = (1+e^{\frac{z}{\tau}})f_\tau(z+i \pi) \,.
\ee
\end{corollary}
The two functional equations~\eqref{eq.diff} and~\eqref{eq.diff2} determine
$f_\tau$ up to multiplication by a doubly periodic function, and
when $\tau >0$ and irrational, such functions are constant; see for example~\cite[p.~251]{Faddeev}.


Our last topic concerns the Stokes phenomenon of the Laplace transform
of $G(\cdot,z)$.
Let $\rho_\th = [0,\infty) e^{i \theta}$ denote the ray in the complex plane
and let 
\be
\label{eq.laptheta}
(\calL^\theta f)(\tau) = \int_{\rho_\theta} e^{-\xi/\tau} f(\xi) \dd\!\xi
\ee
denote the Laplace transform of a function $f$, integrable along $\rho_\theta$.
Recall that the singularities of the meromorphic function $G(\cdot,z)$ are in
an arrangement of lines $L(z)$ through the origin whose complement
$\BC\setminus L(z) = \cup_{m \in \BZ} C_m(z)$ is a union of open cones
\be
\label{eq.Cm}
 \qquad
C_m(z) = \{ \xi \in \BC^* \,|\, \arg(\xi_{m-1}(z)) < \arg(\xi) <
\arg(\xi_{m}(z)) \} \,.
\ee
It follows that when $\theta \in C_m(z)$, the Laplace transform
$(\calL^\theta G)(\tau,z)$ is independent of $\theta$ and defines a
holomorphic function of $\tau$ for 
$\arg(\xi_{m-1}(z)) -\pi/2 < \arg(\tau) < \arg(\xi_m(z)) +\pi/2$.
When $\theta = 0 \in C_0(z)$, Theorem~\ref{thm.2} implies that
$(\calL^\theta G)(\tau,z)$ is, up to a dilogarithm term, equal to
$f_\tau(z)$ and the latter is equal to
  \be
  \label{eq.3f}
  \log f_\tau(z) = \log(-q^{\frac{1}{2}} e^z;q)_\infty - 
  \log(-\tq^{\frac{1}{2}} e^{z/\tau};\tq)_\infty  
  \ee
  when $\Im(\tau)>0$ and $\Re(z)<0$, as follows from
  Equation~\eqref{eq.phiproduct}.
On the other hand, by crossing the walls of $L(z)$, we get
\be
\label{eq.GL2}
(\calL^{\pi/2} G)(\tau,z) - (\calL^{0} G)(\tau,z) =
\sum_{m=0}^\infty f_m(\tau,z) - f_{m+1}(\tau,z)
\ee
where $f_m(\tau,z) = (\calL^{\theta_m(z)} G)(\tau,z)$ is the Laplace
transform of $G(\cdot,z)$ along a ray $\theta_m(z) \in C_m(z)$. When
$\Re(z)<0$ and $|\Im(z)|< \pi$ and $\Im(\tau)>0$, the difference
$f_m(\tau,z) - f_{m+1}(\tau,z)$ is obtained by adding the poles
$-n \xi_{-m-1}(z)$ with $n >0$ of $G(\cdot,z)$ at the corresponding ray.
Using the residue of $G(\cdot,z)$ at these points given in
Theorem~\ref{thm.1} and adding up, it follows that
\begin{align}
  \notag
  f_m(\tau,z)-f_{m+1}(\tau,z) &=
\sum_{n =1}^\infty 2 \pi i C_{-n} e^{\xi_{-m-1}(z)n/\tau}                           
=  -\sum_{n =1}^\infty \frac{(-1)^n}{n} e^{\xi_{-m-1}(z)n/\tau} 
  \\ \label{eq.fstokes}
  &= \log(1+e^{\xi_{-m-1}(z)/\tau})= \log(1+e^{z/\tau} \tq^{m+\frac{1}{2}}) \,.
\end{align}
This, combined with Equations~\eqref{eq.3f} and~\eqref{eq.GL2}, implies
that 
\be
\label{eq.fminus}
(\calL^{\pi/2} G)(\tau,z)
=\log(-q^{\frac{1}{2}} e^z;q)_\infty 
\ee
in agreement with the result of Marmi--Sauzin proven in Sec.~1.1 and Thm.~4.3
of~\cite{MS}.


\section{Proofs}
\label{sec.proofs}

\subsection{The formal power series solution of
  the difference equation}
\label{sub.formal}

The following lemma is well-known and standard
(see eg~\cite[Sec.~13.4, Prop.~6]{AK:TQFT}), but we include its proof for
completeness.

\begin{lemma}
  \label{lem.formal}
If $f_\tau(z)$ satisfies Equation~\eqref{eq.diff} and
$\lim_{z \to -\infty}f_\tau(z)=1$, then for fixed $z$, $f_\tau(z)$
admits an asymptotic expansion of the form~\eqref{eq.Phia}
with $\phih_\tau(z)$ given in~\eqref{eq.phit}.
\end{lemma}

\begin{proof}
  Letting $\phi_\tau(z)=\log f_\tau(z)$, it follows that
  $$
  \phi_\tau(z+\pi i \tau)-\phi_\tau(z-\pi i \tau)=-\log(1+e^z) \,.
  $$
  Taylor's theorem combined with $-\log(1+e^z)=\pt_z \Li_2(-e^z)$
  implies that
  $$
  2 \sinh(\pi i \tau \pt_z) \phi_\tau(z) = \pt_z \Li_2(-e^z) 
  $$
  hence, that
  $$
  2 \pi i\phi_\tau(z) =
  \frac{\pi i \pt_z}{\sinh(\pi i \tau \pt_z)} \Li_2(-e^z) \,.
  $$
  The expansion
  $$
  \frac{z}{\sinh(z)}=\sum_{n=0}^\infty B_{2n}(1/2) \frac{(2z)^{2n}}{(2n)!}
  $$
  concludes the proof of the lemma.
\end{proof}

\subsection{The Borel transform}
\label{sub.borel}

Consider the formal power series
\be
\phi_f(\tau,z)=
\sum_{n=1}^\infty \frac{B_{2n}(1/2)}{(2n)!} f^{(2n)}(z) (2\pi i)^{2n-1}
\tau^{2n-1}
\ee
for a function $f$ analytic on $z$ with $|\Im(z)|<\pi$, and let 
$G_f(\xi,z)=\calB(\phi_f(\cdot,z))$ denote the corresponding Borel transform.

\begin{proposition}
  \label{prop.1}
  We have:
  \be
  \label{eq.bor}
  G_f(\xi,z)=\frac{i}{2\pi } \sum_{n=1}^\infty   
\frac{(-1)^n}{n^2}
\left(f''\left(z+\frac{\xi}{n}\right)+
f''\left(z-\frac{\xi}{n}\right)\right) \,.
  \ee
\end{proposition}

\begin{proof}
  The proof is rather standard. It uses the Hadamard product $\circledast$
  of power series (which was also used in~\cite{CG:EM}) whose definition
  we recall:
\begin{equation}
\label{eq.hadamard}
\left( \sum_{n=0}^\infty b_n \xi^n \right) \circledast
\left( \sum_{n=0}^\infty b_n \xi^n \right)
=
 \sum_{n=0}^\infty b_n c_n \xi^n \,.
\end{equation}

We have:
\begin{align*}
  G_f(\xi,z) &=
\sum_{n=1}^\infty \frac{B_{2n}(1/2)}{(2n)!} 
\frac{f^{(2n)}(z)}{(2n-2)!} (2\pi i)^{2n-1} \xi^{2n-2} \\
&= 
\left(
\sum_{n=1}^\infty \frac{B_{2n}(1/2)}{(2n)!} \xi^{2n-2}
\right) \circledast
\left(
\sum_{n=1}^\infty 
\frac{f^{(2n)}(z)}{(2n-2)!} (2 \pi i)^{2n-1} \xi^{2n-2}
\right) \\
&= f_1(\xi) \circledast f_2(\xi,z)
\end{align*}
where
\begin{align}
\label{eq.f1p}
f_1(\xi) &= \sum_{n=1}^\infty \frac{B_{2n}(1/2)}{(2n)!} \xi^{2n-2}
& 
f_2(\xi,z) &= 2\pi i \sum_{n=1}^\infty 
\frac{f^{(2n)}(z)}{(2n-2)!} (2\pi i \xi)^{2n-2} \,.
\end{align}
Now, since $B_{m}(1/2)=0$ for every odd $m$, we have:
\begin{align*}
  f_1(\xi) &= \sum_{n=1}^\infty \frac{B_{2n}(1/2)}{(2n)!} \xi^{2n-2}
  =
\frac{1}{\xi^2} 
             \sum_{n=1}^\infty \frac{B_{2n}(1/2)}{(2n)!} \xi^{2n}
  =
\frac{1}{\xi^2} 
\sum_{n=1}^\infty \frac{B_{n}(1/2)}{n!} \xi^{n} \\
&= 
\frac{1}{\xi^2}  \left( \frac{e^{\xi/2}\xi}{e^\xi-1}-1 \right) =   
\frac{1}{\xi(e^{\xi/2}-e^{-\xi/2})}-\frac{1}{\xi^2} 
\end{align*}
and Taylor's theorem gives:
\begin{align*}
f_2(\xi,z) &= 
2 \pi i \sum_{n=1}^\infty 
\frac{f^{(2n)}(z)}{(2n-2)!}  (2\pi i\xi)^{2n-2} \\
&= 
(2 \pi i)^{-1} \partial^2_\xi \sum_{n=1}^\infty 
\frac{f^{(2n)}(z)}{(2n)!}  (2\pi i\xi)^{2n} \\
&= 
(2 \pi i)^{-1} \partial^2_\xi
\left( \frac{1}{2} ( f(z+2\pi i\xi)+f(z-2\pi i\xi)) -f(z) \right) \\
&=\pi i 
 \left(f''(z+2\pi i\xi)+f''(z-2\pi i\xi) \right) \,.
\end{align*}
Now, using Cauchy's theorem, it follows that
$$
G_f(\xi,z)=\frac{1}{2 \pi i} 
\int_{\gamma}
f_1(s) f_2 \left(\frac{\xi}{s},z \right) \frac{ds}{s}
$$
where $\gamma$ is a small circle around $0$. The function 
$f_1(s)$ has simple poles at $2 \pi i m$ for $m \in \BZ\setminus \{0\}$ 
with residue $(-1)^m/(2 \pi i m)$.
Now, deform the integration 
contour to circles of increasing radii and collect the residues.
Since $f_1(s)=O(1/s)$ and $f_2(\xi/s)=O(1/s^2)$, it follows that the integrand
is $O(1/s^3)$, thus the contribution from infinity is zero.
The residue of the integrand for $m \in \BZ\setminus \{0\}$ is given by:

\begin{align*}
\Res\left(f_1(s) f_2\left(\frac{\xi}{s},z\right) 
\frac{1}{s}, s=2 \pi i m\right) &=
\frac{1}{2 \pi i m} f_2\left(\frac{\xi}{2 \pi i m},z\right) 
\Res(f_1(s),s=2 \pi i m)
\\
&=
\frac{(-1)^m}{4 \pi i m^2}\left(f''\left(z+\frac{\xi}{m}\right)+
f''\left(z-\frac{\xi}{m}\right)\right) \,.
\end{align*}
Thus, collecting the residues, it follows that

\begin{align*}
G_f(\xi,z) &= -\sum_{m \in \BZ\setminus 0}
\frac{(-1)^m}{4 \pi i m^2}  
\left(f''\left(z+\frac{\xi}{m}\right)+
f''\left(z-\frac{\xi}{m}\right)\right) \\
&= 
\frac{i}{2 \pi } \sum_{m=1}^\infty   
\frac{(-1)^m}{m^2}
\left(f''\left(z+\frac{\xi}{m}\right)+
f''\left(z-\frac{\xi}{m}\right)\right) \,.
\end{align*}
For fixed $z$ and $\xi$, 
the above sum is dominated by $\sum_{m=1}^\infty 1/m^2$ and thus the
convergence is uniform on compact sets. This concludes the proof of
Proposition~\ref{prop.1}.
\end{proof}

\begin{proof}[Proof of Theorem~\ref{thm.1}]
Apply Proposition~\ref{prop.1} to the function $f(z)=\Li_2(-e^z)$
which satisfies
$$
f''(z)=-\frac{1}{1+e^{-z}} \,.
$$
\end{proof}

\subsection{Bounds}
\label{sub.bounds}

In this section we give a proof of Theorem~\ref{thm.anal}.

We begin with the following lemma.

\begin{lemma}
  \label{lem.ez}
  When $z \in \BC$ with $|\Im(z)|<\pi$ we have:
  \be
  \label{eq.ez}
  \frac{1}{|1+e^z|} \leq
  \begin{cases}
  1 & \text{if} \quad \cos(\Im(z)) \geq 0 \\
  \frac{1}{|\sin(\Im z)|} & \text{if} \quad \cos(\Im(z)) \leq 0
  \end{cases}
  \ee
\end{lemma}

\begin{proof}
  With $z=t + i a$, we have:
  $$
  |1+e^z|^2 = e^{2t}+2 e^t \cos a + 1
  $$
  and the right hand side, as a function of $t \in \BR$,
  has critical points in $t_0 \in \BR$ such that
  $e^{t_0}+\cos a =0$. When $\cos a>0$, it follows that
  $\inf_{t \in \BR} |1+e^z|^2 =1$, and when $\cos a \leq 0$, it follows that
  $t_0$ is a global minimum and $e^{2t_0}+2 e^{t_0} \cos a + 1 = \sin^2 a$.
  The result follows.
\end{proof}

\begin{proof}(of Theorem~\ref{thm.anal})
  The first part follows from Equation~\eqref{eq.G}, Lemma~\ref{lem.ez}
  (applied to $\xi/n-z$ for $n \in \BZ^*$) and
  the fact that $\sum_{n=1}^\infty 1/n^2 = \pi^2/6$. In particular, it
  implies that the Laplace transform $(\calL G)(\tau,z)$ is well-defined
  and even extends to $\tau \in \BC$ with $\Re(\tau)>0$. 

  The second part follows from what is known in the literature as
  Watson's lemma~\cite{Watson}, a modern proof of which may be found
  for instance in Miller~\cite[p.~53, Prop.~2.1]{Miller:applied}.
  It is also known
  as the fine Borel--Laplace transform in
  Mitschi--Sauzin~\cite[Sec.~5.7, Thm.~5.20]{Sauzin:divergent} whose proof follows
  Malgrange~\cite{Malgrange:sommation}. Our proof of the second part
  follows directly from 
  Mitschi--Sauzin~\cite[Sec.~5.7, Thm.~5.20]{Sauzin:divergent}, together with
  the upper bound from the first part (which implies that $c_0=0$ and $c_1=\e$
  in the notation of Theorem~5.20 of~\cite{Sauzin:divergent}).
\end{proof}

\subsection{The Laplace transform}
\label{sub.laplace}

Woronowicz~\cite{Woro:quantum}, while studying the quantum exponential
function via functional analysis, introduced the following function
\be
\label{eq.woro}
W_{\theta}(z) =\int_\BR
  \frac{\log(1+e^{\theta \xi})}{1+e^{\xi-z}} \dd\!\xi 
\ee
defined for $\theta >0$ and $|\Im(z)|<\pi$, and proved that (after some
elementary change of variables) it satisfies the functional
equation~\eqref{eq.diff}; see~\cite[Eqn.~(B.3)]{Woro:quantum}. Thus, one can
relate Woronowicz's function with the quantum dilogarithm as is done
in~\cite[Eqns.~(1), (2)]{Kashaev:YB} without  proof. The next proposition
provides a formal proof of this fact.

\begin{proposition}
  \label{prop.woro1}
When $\tau >0$ and $z\in \BC$ with $|\Im(z)|<\pi$, we have:
  \be
  W_{\frac1\tau}(z)=
  -2\pi i\log\fadb\!\left(\frac z{2\pi\sfb}\right) \,,\quad \sfb^{2}=\tau.
  \ee
\end{proposition}

\begin{proof}
  The proof uses a mixture of real and complex analysis. Let $\epsilon$ be a
  positive real number such that $0<\epsilon<\min(1,\theta)$. 
  We remark that  $W_{\theta}(z) $ can be interpreted as a value of the
  scalar product in the complex Hilbert space $ L^2(\BR)$ of square
  integrable functions on the real line with respect to the Lebesgue measure
\be
W_{\theta}(z)=\langle f|g\rangle=\int_{\BR}f(x)\overline{g(x)}\dd\! x
\ee
where $f,g\in L^2(\BR)$ are defined by
\be
f(x)=e^{-\epsilon x}\log(1+e^{\theta x}),\quad g(x)
=\frac{e^{\epsilon x}}{1+e^{x-\bar z}} \,.
\ee
As the Fourier transformation
\be
(F\psi)(x)=\int_{\BR}\psi(y)e^{2\pi i xy}\dd\! y 
\ee
is a unitary operator in $L^2(\BR)$, we have the equality
\be
\langle f| g\rangle=\langle Ff| Fg\rangle.
\ee
 By using Lemma~\ref{lem. fourier_of_inv_cosh} below, 
 we can calculate explicitly the elements $Ff, Fg\in L^2(\BR)$.
 Indeed, denoting $\zeta=2\pi x+i\epsilon$, we have
\begin{align}
  (Ff)(x) &=\int_{\BR}e^{i\zeta y}\log(1+e^{\theta y})\dd\! y=\frac{i\theta}{\zeta}\int_{\BR}\frac{e^{i\zeta y}}{1+e^{-\theta y}}\dd\! y
            \\ 
 &=\frac{2\pi i}{\zeta}\int_{\BR}\frac{e^{2\pi i\zeta y/\theta}}{1+e^{-2\pi y}}\dd\! y
   =\frac{\pi i}{\zeta}\int_{\BR}\frac{e^{(\frac \zeta \theta-\frac i2) 2\pi i y}}{\cosh(\pi y)}\dd\! y
   \\ &=\frac{\pi i}{\zeta\cosh(\frac{\pi\zeta}\theta-\frac{\pi  i}2)} =\frac{\pi i}{\zeta\cos(\frac{\pi }2-\frac{\pi \zeta}{i\theta})} =\frac{\pi i}{\zeta\sin(\frac{\pi \zeta}{i\theta})}=-\frac{\pi }{\zeta\sinh(\frac{\pi \zeta}{\theta})}
\end{align}
where in the second equality we integrated by parts, and 
\begin{align}
  \overline{ (Fg)(x)}&=\int_{\BR}\frac{e^{-i\zeta y}}{1+e^{y-z}}\dd\! y=\int_{\BR-z}\frac{e^{-i\zeta (y+z)}}{1+e^{y}}\dd\! y
                       \\ &=2\pi e^{-i\zeta z}\int_{\BR-\frac{z}{2\pi}}\frac{e^{ -2\pi i\zeta y}}{1+e^{2\pi y}}\dd\! y
                            =\pi e^{-i\zeta z}\int_{\BR-\frac{z}{2\pi}}\frac{e^{(\frac i2-\zeta) 2\pi iy}}{\cosh(\pi y)}\dd\! y
  \\
   &=\frac{\pi e^{-i\zeta z}}{\cosh(\frac {\pi i} 2-\pi\zeta)}
  =\frac{\pi e^{-i\zeta z}}{\cos(\frac {\pi }2 +\pi i\zeta)}=\frac{\pi e^{-i\zeta z}}{\sin(-\pi i\zeta)}=\frac{\pi ie^{-i\zeta z}}{\sinh(\pi \zeta)}
\end{align}
where in the fifth equality we used the condition $|\Im (z)| <\pi$.
Thus, we obtain
\be
W_\th(z)=\langle Ff| Fg\rangle=\int_{\BR} (Ff)(x)\overline{ (Fg)(x)}\dd\! x
=\int_{\BR+i\epsilon}
\frac{-\pi ie^{-i\zeta z}}{2\sinh(\frac{\pi \zeta}{\theta})\sinh(\pi \zeta)}
\frac{\dd\! \zeta}\zeta
\ee
which implies that
\be
\frac i{2\pi}W_{\frac1\tau}(z)=
\int_{\BR+i\pi\sfb\epsilon}
\frac{e^{-i\frac{\zeta z}{\pi\sfb}}}{4\sinh( \sfb\zeta)\sinh( \zeta/\sfb)}
\frac{\dd\! \zeta}\zeta=\log\fadb\!\left(\frac z{2\pi\sfb}\right).
\ee
\end{proof}

The next lemma is well-known, see e.g., Godement~\cite[VII, \&3.15]{Godement:III}.
We will also give a proof using~\cite[Lem.~2.1]{GK:evaluation}.

\begin{lemma}
  \label{lem. fourier_of_inv_cosh}
  When $w,\sigma\in \BC_{|\Im|<\frac12}=\{u\in\BC\mid |\Im(u)|<\frac12\}$,
  we have:
 \be
 \int_{\BR+\sigma}\frac{e^{2\pi iwz}}{\cosh(\pi z)}\dd\! z
 =\frac1{\cosh(\pi w)} \,.
\ee
\end{lemma}

\begin{proof}
  By using \cite[Lem.~2.1]{GK:evaluation} with
  $f(z)=\frac{e^{2\pi iwz}}{\cosh(\pi z)}$ and $a=i$, we have
 \be
 \frac{f(z+a)}{f(z)}=
 \frac{\cosh(\pi z)e^{2\pi iw(z+i)}}{\cosh(\pi z+\pi i)e^{2\pi iwz}}=
 -e^{-2\pi w}
 \ee
 so that
\begin{align}
\int_{\BR+\sigma}\frac{e^{2\pi iwz}}{\cosh(\pi z)}\dd\! z&=\left(\int_{\BR+\sigma}-\int_{\BR+\sigma+i}\right)\frac{e^{2\pi iwz}}{(1+e^{-2\pi w})\cosh(\pi z)}\dd\! z\\
                                                         &=2\pi i \,\Res_{z=\frac i2}\left(\frac{e^{2\pi iwz}}{(1+e^{-2\pi w})\cosh(\pi z)}\right)
  \\
  &=\frac{\pi i}{\cosh(\pi w)} \,\Res_{z=0}\left(\frac1{\sin(\pi iz)}\right)=\frac{1}{\cosh(\pi w)} 
\end{align}
where, in the second equality, the application of the residue theorem is justified by the limits
\be
\lim_{x\to\pm\infty}|f(x+\sigma +it)|\le\left|e^{-2\pi wt}\right|\lim_{x\to\pm\infty}\left(\frac{e^{2\pi |\Im(w)||x+\Re(\sigma)|}}{\sinh|\pi(x+\Re(\sigma))|}\right)=0.
\ee
\end{proof}

The next proposition identifies Woronowicz's formula for the quantum
dilogarithm with the Laplace transform of the function $G(\xi,z)$.

\begin{proposition}
  \label{prop.woro2}
When $\tau>0$ and $z \in \BC$ with $|\Im(z)| < \pi$, we have:
  \be
  \label{eq.woro2}
W_{\frac 1\tau}(z) =
  -\frac{1}{\tau} \Li_2(-e^z) -2\pi i (\calL G)(\tau,z) \,.
  \ee
\end{proposition}

\begin{proof}
  We have:
  \begin{align*}
    \label{eq.woro2}
 W_{\frac 1\tau}(z)& =   \int_{-\infty}^\infty \frac{\log(1+e^{\xi/\tau})}{e^{\xi-z}+1} \dd\!\xi
    =
      \int_{-\infty}^0 \frac{\log(1+e^{\xi/\tau})}{1+e^{\xi-z}} \dd\!\xi +
      \int_0^\infty \frac{\log(1+e^{\xi/\tau})}{1+e^{\xi-z}} \dd\!\xi                    \\
&=    
                                                                                            \int_0^\infty \left( \frac{\log(1+e^{\xi/\tau})}{1+e^{\xi-z}}+
                                                                                            \frac{\log(1+e^{-\xi/\tau})}{1+e^{-\xi-z}}\right)  \dd\!\xi
    \\
    &=
      \frac1\tau\int_0^\infty \frac{\xi}{1+e^{\xi-z}}\dd\!\xi +
      \int_0^\infty
      \log(1+e^{-\xi/\tau})\left(\frac{1}{1+e^{\xi-z}}+ \frac{1}{1+e^{-\xi-z}}\right)\dd\!\xi \\
    &=
      \frac1\tau\int_0^\infty \frac{\xi}{1+e^{\xi-z}}\dd\!\xi-
      \int_0^\infty \sum_{n=1}^\infty \frac{(-1)^n}{n} e^{-n\xi/\tau}
 \left(\frac{1}{1+e^{\xi-z}}+ \frac{1}{1+e^{-\xi-z}}\right)\dd\!\xi   
  \end{align*}
where the last equality follows from expanding the logarithm. 
  Rescaling $\xi \to \xi/n$ and using the identity
\be
  \int_0^\infty \frac{\xi}{1+e^{\xi-z}}\dd\!\xi = -\Li_2(-e^z)
\ee
  (which can be verified for instance by integrating by parts)
  it follows that
\begin{align*}
  W_{\frac 1\tau}(z) +
  \frac{1}{\tau} \Li_2(-e^z)&= - \int_0^\infty e^{-\xi/\tau}
  \sum_{n=1}^\infty \frac{(-1)^n}{n^2} 
  \left(\frac{1}{1+e^{\frac\xi n-z}}+ \frac{1}{1+e^{-\frac\xi n-z}}\right)\dd\!\xi\\
  &=-2\pi i \int_0^\infty e^{-\xi/\tau}G(\xi,z)\dd\!\xi =-2\pi i(\calL G)(\tau,z)
\end{align*}
  where we used Equation~\eqref{eq.G}. 
\end{proof}

We are now ready to give a proof of Theorem~\ref{thm.2}.
\begin{proof}(of Theorem~\ref{thm.2})
  Propositions~\ref{prop.woro1} and Equation~\eqref{eq.woro2} imply that
  \begin{equation}
   -2\pi i   \log \fadb\left(\frac z{2\pi \sfb}\right) = W_{\frac1\tau}(z)= 
  -\frac{1}{\tau} \Li_2(-e^z) -2\pi i (\calL G)(\tau,z)   
  \end{equation}
  which is equivalent to \eqref{eq.thm2}.
     \end{proof}


\section{Useful properties of the dilogarithm function}
\label{sec.fad}

In this section we collect some useful properties of the quantum
dilogarithm function

\begin{equation}
  \label{eq.phiproduct}
\fadb(z)
=\frac{(e^{2 \pi b (z+c_\sfb)};q)_\infty}{
(e^{2 \pi b^{-1} (z-c_\sfb)};\tq)_\infty}
\end{equation}
where
\be
\label{eq.qqt}
q=e^{2 \pi i \sfb^2}, \qquad 
\tq=e^{-2 \pi i \sfb^{-2}}, \qquad
c_b=\frac{i}{2}(\sfb+\sfb^{-1}), \qquad \Im(\sfb^2) >0 \,.
\ee
The above function can also be defined in the lower half-plane
$\Im(\sfb^2) <0$ using the symmetry
\be
\label{eq.fsym}
\fadb(z)=\fad{\sfb^{-1}}(z) \,,
\ee
and remarkably the function of $\sfb^2 \in \BC\setminus\BR$
admits an extension to $\sfb^2 \in \BC'=\BC\setminus(-\infty,0]$.
The integral representation~\eqref{eq.fad} implies the additional
symmetry
\be
\label{eq.fsym2}
\fadb(z)=\fad{-\sfb}(z) \,.
\ee
$\fadb(z)$ is a meromorphic function of $z$ with
$$
\text{poles:} \,\,\, c_\sfb + i \BN \sfb + i \BN \sfb^{-1},
\qquad
\text{zeros:} \,\, -c_\sfb - i \BN \sfb - i \BN \sfb^{-1} \,.
$$
It satisfies the inversion relation
$$
\fadb(z) \fadb(-z)=e^{\pi i z^2} \fadb(0)^2, 
\qquad
\fadb(0)=q^{\frac{1}{24}} \tq^{-\frac{1}{24}} \,.
$$
It is a quasi-periodic function satisfying 
\be
\label{eq.qp}
\fadb(z-i \sfb/2) =
(1+e^{2 \pi \sfb z})\fadb(z+i \sfb /2)
\ee
which, due to the symmetry~\eqref{eq.fsym}, implies a second functional
equation
\be
\label{eq.qp2}
\fadb(z+i \sfb/2) =
(1+e^{-2 \pi \sfb z})\fadb(z-i \sfb /2)
\ee
obtained from~\eqref{eq.qp} by replacing $\sfb$ by $\sfb^{-1}$. Note that
Equation~\eqref{eq.qp} (resp.~\eqref{eq.qp2}) implies that the function
$\fadb(z/(2\pi \sfb))$ satisfies Equation~\eqref{eq.diff}
(resp.~\eqref{eq.diff2}) with $\tau=\sfb^2$.


%


\section*{Acknowledgments}

Theorems~\ref{thm.1} and~\ref{thm.2} were part of an unpublished manuscript
from 2006, written during a visit of the first author to Geneva.
The authors wish to thank the University of Geneva and the International 
Mathematics Center at SUSTech University, Shenzhen for their hospitality.
The results of the paper were presented in a Resurgence Conference in Miami
in 2020. The authors wish to thank the organizers for their hospitality.

S.G. wishes to thank David Sauzin for enlightening
conversations and for a careful reading of the manuscript.

This work is partially supported by the Swiss National Science Foundation
research program NCCR  The Mathematics of Physics (SwissMAP) and the ERC
Synergy grant Recursive and Exact New Quantum Theory (ReNew Quantum).


\bibliographystyle{hamsalpha}
\bibliography{biblio}
\end{document}